\documentclass[11pt]{article}
\usepackage{amsmath, amssymb, amsthm}
\usepackage{amsmath, amssymb}
\usepackage{geometry}

\usepackage{times}

\newcommand{\eps}{\epsilon}

\newcommand{\EE}{\mathcal{E}}
\newcommand{\QQ}{\mathcal{Q}}
\newcommand{\CC}{\mathcal{C}}

\def\argmin{\mathop{\rm argmin}}

\def\poly{\mathop{\rm poly}}
\newcommand\TOPSTEAL{\operatorname{TOP-STEAL}}
\newcommand\dq{\operatorname{dq}}

\newtheorem{theorem}{Theorem}[section]

\newtheorem{corollary}[theorem]{Corollary}

\newtheorem{lemma}[theorem]{Lemma}
\newtheorem{proposition}[theorem]{Proposition}

\newtheorem{remark}[theorem]{Remark}

\theoremstyle{remark}

\newtheorem{claim}[theorem]{Claim}

\theoremstyle{definition}
\newtheorem{definition}[theorem]{Definition}
\newcommand{\HuNote}[1]{}
\newcommand{\hlhf}[1]{#1}

\begin{document}
\title{On the Complexity of Computing an Equilibrium in\\ Combinatorial Auctions}

\author{Shahar Dobzinski\thanks{Weizmann Institute of Science.} \and Hu Fu\thanks{Microsoft Research.} \and Robert Kleinberg\thanks{Cornell University.}}

\maketitle

\begin{abstract}
We study combinatorial auctions where each item is sold separately but
simultaneously via a second price auction. We ask whether it is possible to
efficiently compute in this game a pure Nash equilibrium with social welfare
close to the optimal one.

We show that when the valuations of the bidders are submodular, 
in many interesting settings (e.g., constant number of bidders, budget additive bidders) computing an equilibrium with good welfare is essentially as easy as computing, completely ignoring incentives issues, an allocation with good welfare. On the other hand, for subadditive valuations, we show that computing an equilibrium requires exponential communication. Finally, for XOS (a.k.a. fractionally subadditive) valuations, we show that if there exists an efficient algorithm that finds an equilibrium, it must use techniques that are very different from 
our current ones.
\end{abstract}

\thispagestyle{empty}\maketitle\setcounter{page}{0}\newpage

\section{Introduction}


Combinatorial auctions have received much attention in recent years. The literature 
is quite large, but roughly speaking it is fair to say that most papers take an engineering-like approach, by designing algorithms that find allocations with high welfare (e.g., \cite{Feige09, V08}), or designing algorithms that achieve good approximation ratios when bidders play their dominant strategies.

Recently, several papers (e.g., \cite{BhawalkarR11,CKS10,FFGL13,HassidimKMN11,LB10,PLST12,R12,ST12,ST13})  took a more ``existential'' approach to analyzing combinatorial auctions. Instead of designing algorithms that give a specific recipe for computing an efficient allocation, a simple game is defined; usually it is assumed that each item is sold separately but simultaneously via some kind of an auction. The implicit assumption is that the players somehow reach an equilibrium, and the quality of this equilibrium is analyzed. The literature considers several variants of such games that differ in the single item auction type (first or second price), and in the solution concept (pure nash, mixed nash, etc.).

The goal of the current paper is to somewhat narrow the gap between these two mindsets. That is, given such a game we want to determine whether we can compute a good equilibrium efficiently.


\vspace{0.1in}\noindent \textbf{The Setting.} In a combinatorial auction there is a set $M$ of items ($|M|=m$) for sale. There is
also a set $N$ of bidders ($|N|=n$). Each bidder $i$ has a valuation function
$v_i:2^M\rightarrow \mathbb R$ that denotes the value of bidder $i$ for every
possible subset of the items. We assume that the valution functions are monotone
(for all $S\subseteq T$, $v_i(S)\geq v_i(T)$) and normalized
($v_i(\emptyset)=0$). The goal is to find an allocation of the items
$(S_1,\ldots, S_n)$ that maximizes the social welfare: $\Sigma_iv_i(S_i)$. We
are interested in algorithms that run in time $\poly(m,n)$. 

If each valuation is naively represented by $2^m$ numbers, reading the input alone will be too slow for our algorithms. Therefore, the standard approach assumes that the valuations belong to some subclass that can be succinctly represented, or that the valuations are represented by black boxes that can only 
     answer
specific types of queries. The standard queries are value query (given a bundle $S$, what is the value of $v(S)$?) and demand query (given $p_1,\ldots, p_m$, what is a bundle that maximizes $v(S)-\Sigma_{j\in S}p_j$?). For impossibility results we will also consider the communication complexity model, where the queries are not restricted and we only count the number of bits communicated. See \cite{BlumrosenNisan07} for a thorough description of these models.

Christodoulou, Kovacs, and Schapira \cite{CKS10} were the first to analyze the quality of such equilibria in combinatorial auctions. They study simultaneous second price auctions. In this game the strategy of each bidder is to bid one number $b_i(j)$ for every item $j$. Player $i$ gets item $j$ if $b_i(j)\geq b_{i'}(j)$ for all players $i'$ 
     (if there are several players with the maximal bid for $j$, then $j$ is given to some arbitrary one of them). 
When player $i$ wins the set of items $S_i$, his payment is $\Sigma_{j\in S_i}\max_{i'\neq i}b_{i'}(j)$. To avoid trivial equilibria, they use the standard no overbidding assumption 
     which
states that for every player $i$ and bundle $S$ we have that $v_i(S)\geq \Sigma_{j\in S}b_i(j)$. See 
     \cite{BhawalkarR11,CKS10} 
for a discussion.


\vspace{0.1in}\noindent \textbf{Our Goal: The Computational Efficiency of Finding a Good Equilibrium.} In this paper we study the basic task of computing a pure equilibrium when the auction format is a simultaneous second price auction. We start with discussing submodular valuations, where for every item $j$ and bundles $S\subseteq T$ we have that $v_i(j|S)\geq v_i(j|T)$.%
\footnote{We use the standard notation $v(A|B)$
to denote $v(A \cup B) - v(B)$\hlhf{, the marginal value of bundle $A$ given
bundle~$B$}.}
\hlhf{A special class of submodular valuations also of interest for us
is the \emph{budget additive valuations}, where the valuation is fully
described by a budget~$b$ and the value of each single item, with $v(S)$ given by $\min \{
b, \sum_{j \in S} v(j)\}$. }

In \cite{CKS10} it is proved that if all valuations are submodular then the (pure) price of anarchy is at most $2$. They even show that this bound can be achieved constructively: the greedy algorithm of \cite{LLN06} 
     --- 
coupled with appropriately chosen prices 
     --- 
actually finds an equilibrium that necessarily has a price of anarchy of $2$. As
better algorithms are known for the general case (an $\frac e
{e-1}$-approximation algorithm \cite{V08}) and for some interesting special
cases (e.g., a $\frac 4 3$-approximation algorithm for \hlhf{budget} additive
bidders \cite{CG10}, and an FPTAS for a constant number of \hlhf{budget} additive bidders \cite{AM04}), they ask whether we can find an equilibrium with an approximation ratio better than $2$ in polynomial time.

Our first set of results partially answers this question by providing a series of black-box reductions that show that in several interesting cases we can exploit an arbitrary approximation algorithm to efficiently find an equilibrium that provides the same (or better) approximation ratio:

\vspace{0.1in}\noindent \textbf{Theorem: } 
Suppose that we are in one of the following settings: (1) a constant number of
submodular bidders, or (2) an arbitrary number of \hlhf{budget} additive bidders. 
Then, 
    given any allocation $(S_1,\ldots,S_n)$,
using only polynomially many value queries it is possible to find a no-overbidding equilibrium with welfare at least $\Sigma_iv_i(S_i)$.

\vspace{0.1in}\noindent 
In fact, we show that the portion of the theorem that deals with a constant number of submodular bidders can be extended to the more general setting of {\em bounded competition}, where the number of bidders may be arbitrary, but the number of bidders that compete for any single item is at most a constant. (See the definition of $t$-restricted instances in Section~\ref{sec:submodular}.)
\hlhf{As an additional consequence of our analysis, we} observe that for general submodular valuations one can modify any $\alpha$-approximation algorithm to compute an equilibrium that provides an 
     $\alpha$-approximation, 
at the cost of an additional 
pseudo-polynomial number of value queries.

Next, we discuss supersets of submodular valuations such as subadditive valuations, where
for every $S$ and $T$, $v(S)+v(T)\geq v(S\cup T)$). For this class, in \cite{BhawalkarR11} it was observed that sometimes there is no pure equilibrium at all. We show that even if equilibrium is known to exist, it is impossible to efficiently find it:

\vspace{0.1in}\noindent \textbf{Theorem: } It takes exponential communication to find a pure no-overbidding equilibrium in combinatorial auctions with subadditive bidders, even if such equilibrium is known to exist.

\vspace{0.1in}\noindent This is the first computational impossibility result for finding a pure equilibrium in combinatorial auctions.\footnote{Related is \cite{DNO14} which shows that a Bayesian equilibrium is hard to find if the distributions are correlated. Independently from and concurrently with our work, Cai and Papadimitriou \cite{CP14} showed computational complexity of Bayesian Nash in similar settings.  Our work focuses on pure Nash and gives both algorithmic results and \emph{communication} complexity lower bounds.  We consider the two works incomparable and complementing each other.}

Of particular interest is the class of XOS valuations that lies between submodular and subadditive valuations in the complement free hierarchy \cite{LLN06}. In \cite{CKS10} it was shown that the price of anarchy is at most $2$ also if the valuations are XOS. They furthermore provide a natural dynamic that always finds an equilibrium, but the communication complexity of this dynamic is exponential.

Whether this dynamic can be altered to converge in polynomial time is unknown.
In fact, in \cite{CKS10} even the more basic question of whether this dynamic
converges in polynomial time if the valuations are known to be submodular (and
not just XOS) was mentioned as open. We provide a negative answer to this question by
showing an instance of combinatorial auctions with two submodular bidders in
which the dynamic takes exponential time to end. This example leads to
a more general impossibility result --- in fact our most technically involved
result --- we therefore take this opportunity to dive into a more technical discussion.


\vspace{0.1in}\noindent\textbf{The Surprising Flexibility of XOS Oracles.} Recall the definition of an XOS valuation. A valuation $v$ is called XOS if
there exist additive valuations (clauses) $a_1,\ldots, a_t$ such that
$v(S)=\max_ka_k(S)$. An \emph{XOS oracle} receives as input a bundle~$S$ and returns the \emph{maximizing clause} of~$S$ --- an additive valuation $a_k$ from $a_1, \ldots, a_t$, such that $v(S)=a_k(S)$.

Let us now describe the exponential time algorithm 
      of \cite{CKS10}
for finding 
     an equilibrium with
XOS valuations. For simplicity, assume that there are only two players. Start with some arbitrary allocation $(S_1,S_2)$. The algorithm queries the XOS oracle of player $1$ for the maximizing clause $a$ of $S_1$. Player $1$ bids $b_1(j)=a(j)$ for each $j\in S_1$ and $b_1(j)=0$ otherwise. Player $2$ now calculates his demand $S'_2$ at price $b_1(j)$ for each item $j$ and is allocated $S'_2$. The algorithm continues similarly: we query the XOS oracle of player $2$ for the maximizing clause $a'$ of $S'_2$. Player $2$ bids $b_2(j)=a'(j)$ for each $j\in S_2$ and $b_2(j)=0$ otherwise. Player $1$ now calculates his demand $S'_1$ at price $b_2(j)$ for each item $j$ and is allocated $S'_1$.

This is essentially a game with two players that are both following best reply strategies. Will this process ever end? First, it is not hard to prove that at every step the sum of prices 
    (i.e., winning bids)
goes up. Since the sum of prices of an allocation is at most the welfare of an allocation (by the definition of XOS) and since the number of all bundles is finite, the process ends after finitely many steps. The proof that the algorithm ends with an allocation that is a 
     two-approximation 
essentially follows from the earlier algorithm of \cite{DNS05} that is identical to the exponential-time algorithm above except that in the algorithm of \cite{DNS05} each player responds exactly once. As noted above, the dynamic may take exponentially many steps to end. However, we observe that the dynamic is not well defined as there are many possible different XOS oracles for a single valuation $v$, each may lead to a different path; some paths may potentially end after polynomially many steps.

To see that there are many possible different XOS oracles for a single valuation $v$, consider a submodular valuation $v$ (recall that every submodular valuation is also XOS). It is known \cite{DNS05} that the following algorithm finds a maximizing clause $a$ of a bundle $S$: arbitrarily order the items, and rename them for convenience to $1,2,\ldots, |S|$. Now for every $j\in S$, let $a(j)=v(j|\{1,\ldots, j-1\})$. Observe that different 
     orderings 
of the items 
     result 
in different maximizing clauses.

This observation may seem useless, but in fact our positive results for submodular bidders are all based on it. More generally, our algorithms can be seen as a variant of the best-reply algorithm of \cite{CKS10}, where instead of best-reply strategies we use better-reply strategies.%
\footnote{A strategy is a \emph{better reply} for a player if it increases his utility over the current strategy that he is playing (but does not necessarily maximize it).} We show that specific implementations of the XOS oracle can guarantee a fast termination of the better-reply 
     algorithm. (The implementation of the XOS oracle also depends on the valuations of the other players.)

We are unable to extend our algorithms for submodular bidders to XOS bidders,
nor to prove that finding an equilibrium with XOS
bidders requires exponential communication. However, we do observe
that all \hlhf{existing} algorithms for general XOS valuations that use XOS oracles (including pure approximation algorithms that do not take incentives issues into account, e.g., \cite{DS06}) work with \emph{any} XOS oracle and do not assume a specific implementation\footnote{In a sense, every reasonable algorithm must work with any implementation. Otherwise, for example, the choice of an ``unfair'' implementation as to which clause to return among several possibilities may be correlated with global information on the valuation, and hence may convey ``illegal'' information.}. For this kind of algorithms we can prove an impossibility result.

\vspace{0.1in}\noindent \textbf{Definition: } A no-overbidding equilibrium $(S_1, S_2)$ is called \emph{traditional} with respect to some XOS oracles $\mathcal O_{v_1}$ and $\mathcal O_{v_2}$ if for each bidder $i$ and item $j$, if $j\in S_i$ then $b_i(j)$ equals the price of $j$ in $\mathcal O_{v_i}(S_i)$, and if $j\notin S_i$ then $b_i(j)=0$.

\vspace{0.1in}\noindent In other words, an equilibrium is called traditional if the prices are consistent with the XOS oracles of the players. Notice that the equilibrium obtained by the algorithm of \cite{CKS10} is indeed traditional.

\vspace{0.1in}\noindent \textbf{Theorem: } Let $A$ be a deterministic algorithm that always produces a traditional equilibrium with respect to some XOS oracles $\mathcal O_{v_i}$. If $A$ is allowed to make only XOS queries to the oracles $\mathcal O_{v_i}$'s in addition to demand and value queries, then, $A$ makes an exponential number of queries. 

\vspace{0.1in}\noindent The proof is inspired by a proof of \cite{NisanBlog} for the hardness of finding an equilibrium in games using only queries that are analogous to value queries. Our proof differs in several technical aspects: first, it holds also for the stronger and more complicated demand queries and not just value queries (for that we use techniques from \cite{BDO12}). Second, the additional structure of our setting implies a more subtle construction, e.g., we have to prove a new isoperimetric inequality for odd graphs. 




\section{Algorithms for Bidders with Submodular Valuations}
\label{sec:submodular}

We now provide algorithms that find good \hlhf{equilibria} when the valuations are submodular. Our goal is to take an allocation and convert it with ``little'' computational overhead to an equilibrium that has at least the same welfare as the initial allocation. 
\hlhf{We first present a generic process that will be our main workhorse,
after which we show that}
different implementations of this process allow us to compute an equilibrium
with little computational overhead for several different settings. In
particular, this will allow us to prove the main result of this section: one can
take any allocation of combinatorial auctions with a constant number of
submodular bidders and find an equilibrium with at least the welfare of the given
allocation, using only polynomially many value queries. We note that a similar auction was considered in \cite{FKL12} in a slightly different context.

\vspace{0.075in}\noindent\textbf{The Iterative Stealing Procedure}
\begin{enumerate}
\item Start with an arbitrary allocation $(S_1, S_2, \ldots, S_n)$.

\item\label{step-update} Each bidder $i$ arbitrarily orders the items in $S_i$.
For every $j\in S_i$, \hlhf{set $b_i(j)=v_i(j|\{1,...,j-1\}\cap S_i)$}. For $j\notin S_i$ let $b_i(j)=0$.


\item\label{step-steal} If there \hlhf{exist} some players $i,i'$ and item $j\in
S_{i'}$ such that $v_i(j|S_i)>b_{i'}(j)$ then let $S_{i'}=S_{i'}-\{j\}$ and
$S_i=S_{i}+\{j\}$. We say that $i$ \emph{steals} the item from $i'$. Return to
Step (\ref{step-update}). If there are no such players $i,i'$ and no item $j$, the process ends.
\end{enumerate}

\begin{remark}
It is interesting to note the relationship of the iterative stealing procedure to the exponential-time algorithm of \cite{CKS10} (which was 
described in the introduction). For simplicity assume that there are only two
bidders. Step~\ref{step-update} can be essentially seen as a call to an XOS
oracle \cite{DNS05}. In the XOS algorithm each bidder ``steals'' from the other
\hlhf{bidders} the set of \hlhf{items} that maximizes his profit, i.e., best-responds
given the bids of the others. In contrast, in our procedure whenever a player
steals an item he essentially plays a ``better reply'', i.e. plays a strategy better than his current one. To
see that, observe that since the valuations are submodular, if there exists a
set of items that a player can steal to maximize the profit, \hlhf{then} there
exists a single item that the player can steal to improve his profit.  Note that
an XOS valuation does not necessarily have this property.
\end{remark}

To give a full specification of the process, one has to specify how each player
orders his items (Step~\ref{step-update}) and which player steals which item if
there are multiple possible steals in Step \ref{step-steal}. The implementation
of Step \ref{step-steal} will turn out to be less important to us, and we will
focus on considering different implementations of Step \ref{step-update}. Before
\hlhf{this},
we show that stealing can only improve the welfare:

\begin{claim}\label{claim-increase-welfare}
Every steal increases the welfare.
\end{claim}

\begin{proof}
Consider bidder $i$ stealing item $j$ from bidder $i'$. Let $(S_1,\ldots,S_n)$
be the allocation before the steal and $(S'_1,\ldots, S'_n)$ be the allocation
after the steal. We show that $v_i(S_i)+v_{i'}(S_{i'}) > v_i(S'_i)+v_{i'}(S'_{i'})$ which suffices to prove the claim since the allocation of bidders other than $i,i'$ did not change. To see this, observe that 
$v_i(S'_i)+v_{i'}(S'_{i'}) = v_i(S_i) + v_i(j|S_i) + v_{i'}(S_{i'}) - v_{i'}(j|S_{i'}-\{j\})$.   

By submodularity $v_{i'}(j|S_{i'}-\{j\})\leq b_{i'}(j)$ (since
$b_{i'}(j)$ is the marginal value of $j$ given some subset of $S_{i'}-\{j\}$).
On the other hand, since bidder $i$ steals item $j$, $v_i(j|S_i)> b_{i'}(j)$,
and therefore \hlhf{$v_i(j|S_i) > v_{i'}(j|S_{i'}-\{j\})$}. This finishes the proof.
\end{proof}


Next we show that the iterative stealing procedure ends in pseudo-polynomial time for any implementation of Step (\ref{step-update}) . We then provide two specific
implementations of Step (\ref{step-update}): in one, if all valuations
are \hlhf{budget} additive then regardless of the initial allocation the
 procedure terminates in polynomial time, and in the other the number of steals is at most $m^n$. We get that:

\begin{theorem} There exist implementations of the iterative stealing procedure such that: 
\begin{itemize}
\setlength{\itemsep}{0pt}
\item If there exists an $\alpha$-approximation algorithm for combinatorial
auctions with submodular bidders then a no-overbidding equilibrium for the
simultaneous second-price auction with the same approximation guarantee can be found \hlhf{by} running the approximation algorithm and then an additional pseudo-polynomial number of value queries.

\item If there is an $\alpha$-approximation algorithm for combinatorial auctions
with \hlhf{budget} additive bidders then a no-overbidding equilibrium with the same approximation guarantee can be found \hlhf{by} running the approximation algorithm and an additional polynomial number of value queries.

\item If there exists an $\alpha$-approximation algorithm for combinatorial
auctions with submodular bidders then a no-overbidding equilibrium with the same approximation guarantee can be found \hlhf{by} running the approximation algorithm and then
an additional $\poly(m,n)\cdot m^n$ value queries.
\end{itemize}
\end{theorem}

Next we prove the third bullet, and the interested reader can look for the proof of the first two bullets in the appendix. Notice that it is clear that the number of value queries that the iterative stealing procedure makes is $poly(m,n)\cdot(\hbox{number of steals})$.



\vspace{0.1in}\noindent\textbf{Constant Number of Submodular Bidders.} The main result of this section concerns
combinatorial auctions with a constant number of submodular bidders --- or, more
generally, those in which the number of bidders competing for any
single item is bounded above by a constant.
To get some intuition, it will
be instructive to consider instances with two bidders. Let $(S_1,S_2)$ be
some allocation. Set bids $b_i(j)$ as in the process using some
arbitrary order. If we happen to arrive at equilibrium, then we are already
done. Else, call bidder $i$ a top competitor of \hlhf{item~$j$ if}
$v_i(\{j\})\geq v_{i'}(\{j\})$, where $i'$ is the other bidder. If there is some
bidder $i$ and item $j\in S_i$ such \hlhf{that} $i$ is a top competitor for $j$,
then if we use an order where $j$ is first, player $i'$ does not want to steal
item $j$ from player $i$. Hence we can imagine that player $i$'s valuation is
$v_i(\cdot |\{j\})$ (which is still submodular), and look for \hlhf{an} equilibrium using the set of items $M-\{j\}$.

Suppose now that all items are not held by their top competitors.
Since we are not at equilibrium, some player $i$ can steal some item $j$
\hlhf{from} the other player $i'$. Observe that $i$ is a top competitor for item
$j$ (since either bidder $i'$ or bidder $i$ is the top competitor of $j$ and we
assumed that $i'$ is not the top competitor). Hence, after stealing the item we
are back to the previous case: bidder $i$ puts the item first in the order, and
we recurse again.  It is not too hard to see that the process ends after at most
$m$ steals with an equilibrium. This will be rigorously shown in
Lemma~\ref{lemma-t=2}.

We start the formal description of the implementation with a few definitions. Fix valuations $v_1,\ldots, v_n$. For every item $j$ let its \emph{competitors} be $C_j=\{i|v_i(\{j\})>0\}$. Bidder $i$ is a \emph{top competitor} for item $j$ if $v_i(\{j\})\geq v_{i'}(\{j\})$, for every other bidder $i'$. An instance is \emph{$t$-restricted} if for every item $j$ we have that $|C_j|\leq t$. Notice that every instance is $n$-restricted.
\hlhf{We will use $v_{i, |S}$ to denote the valuation $v_i(\cdot |S)$, and the terms ``prices'' and ``bids'' interchangeably when there is no confusion.}

We now give a recursive implementation of the process for $t$-restricted
instances with $n$ bidders and $m$ items. Let the maximum number of steals
that this procedure makes be $f_m(t)$. We show that $f_m(t)\leq
1+f_{m}(t-1)+f_{m-1}(t)$, $f_m(2)\leq m$ for every $m$, and that $f_1(t)\leq 1$
for every $t$. Hence, $f_m(n)\leq m^{n-1}$, as
needed.%
\footnote{Formally, one can prove by induction that $f_m(t)\leq \binom{m+t-1}{t-1} - 1$.}  
We now describe the process, \hlhf{$\TOPSTEAL((v_1,\ldots, v_n), (S_1, \cdots, S_n), M, t)$}. The procedure takes in a profile of valuations $(v_1, \ldots, v_n)$, an initial allocation $(S_1, \cdots, S_n)$ on a set of items~$M$ where each item is $t$-restricted, and returns an allocation $(S'_1,\ldots,S'_n)$ and \hlhf{bids $\{b_i(j)\}$} for each bidder~$i$ and item~$j$. We show that it produces a no-overbidding equilibrium whose social welfare is at least that of the initial allocation.  The implementation itself consists of the definition of $\TOPSTEAL$ in three disjoint cases.

In every procedure we assume that no
item belongs to a non-competitor in the initial allocation (we can check this condition and move items from non-competitors to competitors if necessary -- this only increases welfare). Observe that items with no competitors can be ignored.

\noindent\makebox[\linewidth]{\rule{8cm}{0.4pt}}
\noindent\textbf{Procedure $\TOPSTEAL((v_1,\ldots, v_n), (S_1, \cdots, S_n),M,t)$ (for $|M|=1$)}

If the item does not belong to a top competitor, then one top competitor steals the item.

\noindent\makebox[\linewidth]{\rule{8cm}{0.4pt}}

\noindent\textbf{Procedure \hlhf{$\TOPSTEAL((v_1,\ldots, v_n), (S_1, \cdots, S_n), M,t)$} (for $|M|>1$, $t=2$)}
\begin{enumerate}
\setlength{\itemsep}{0pt}
\item \emph{If there \hlhf{exist an} item $j$ and \hlhf{a} bidder $i$ such that
\hlhf{$j\in S_i$} and $i$ is \hlhf{a} top competitor for $j$}:
\begin{enumerate}
\setlength{\itemsep}{0pt}
\item Let $(S'_1,\ldots,S'_n)$ and \hlhf{$\{b'_{i'}(j')\}_{i' \in N, j' \in M}$} be the allocation and prices
returned by\\ \hlhf{ $\TOPSTEAL((v_{i, |\{j\}}, v_{-i}), (S_i \setminus \{j\},
S_{-i}),  M\setminus\{j\},t)$}. 
\item Let $b_{i'}(j')=b'_{i'}(j')$ for every $j'\neq j$.  Let
$b_i(j)=v_i(\{j\})$ and $b_{i'}(j)=0$ \hlhf{for each} $i'\neq i$. Return
$(S'_i\cup \{j\},S'_{-i})$ and the prices \hlhf{$\{b_i(j)\}_{i \in N, j \in M}$}.
\end{enumerate}

\item \label{step:no-top-comp-t=2}
\emph{Else, there is no item $j$ and bidder $i$ such that \hlhf{$j\in
S_i$} and $i$ is \hlhf{a} top competitor for $j$:}
\begin{enumerate}
\setlength{\itemsep}{0pt}
\item Each bidder $i$ arbitrarily orders his items. For every $j \in S_i$, set $b_i(j)=v_i(j|\{1,...,j-1\}\cap S_i)$. For $j\notin S_i$, set $b_i(j)=0$.
\item If $(S_1,\ldots, S_n)$ and the prices \hlhf{$\{b_i(j)\}_{i \in N, j \in
M}$} are an equilibrium, return them. 
\item Otherwise, find some bidder $i$ that can steal some item $j$ from bidder
$i'$, such that $i$ is a top competitor for~$j$.
\item Define $(S'_1,\ldots,S'_n)$ to be the allocation after $i$ steals $j$ from $i'$; thus $S'_k = S_k$ for $k \neq i,i'$, $S'_i = S_i \cup \{j\}$, $S'_{i'} = S_{i'} \setminus \{j\}$.
\label{step:no-top-comp-t=2-steal}
\item Return: \\ $\TOPSTEAL((v_1, \ldots, v_n),(S'_1, \ldots, S'_n),M,t)$.
\label{step:no-top-comp-t=2-recurse}
\end{enumerate}
\end{enumerate}
\noindent\makebox[\linewidth]{\rule{8cm}{0.4pt}}
\noindent\textbf{Procedure \hlhf{$\TOPSTEAL((v_1,\ldots, v_n), (S_1, \cdots, S_n), M,t)$} (for $|M|>1$, arbitrary $t> 2$)}
\begin{enumerate}
\setlength{\itemsep}{0pt}
\item \emph{If there \hlhf{exist an} item $j$ and \hlhf{a} bidder $i$ such that
\hlhf{$j\in S_i$} and $i$ is \hlhf{a} top competitor for $j$}:
\begin{enumerate}
\setlength{\itemsep}{0pt}
\item Let $(S'_1,\ldots,S'_n)$ and \hlhf{$\{b'_{i'}(j')\}_{i' \in N, j' \in M}$}
be the allocation and prices returned by\\ \hlhf{$\TOPSTEAL((v_{i, |\{j\}},v_{-i}), (S_i \setminus \{j\}, S_{-i}), M \setminus \{j\},t)$}.  
\item Let $b_{i'}(j')=b'_{i'}(j')$ for every $j'\neq j$. Let $b_i(j)=v_i(\{j\})$
and $b_{i'}(j)=0$ \hlhf{for each} $i'\neq i$. Return
$(S'_i\cup\{j\},S'_{-i})$ and the prices \hlhf{$\{b_i(j)\}_{i \in N, j \in M}$}.
\end{enumerate}

\item \emph{Else, there is no item $j$ and bidder $i$ such that \hlhf{$j\in
S_i$} and $i$ is \hlhf{a} top competitor for $j$:}
\label{step:no-top-comp-t>2}
\begin{enumerate}
\setlength{\itemsep}{0pt}
\item For each bidder $i$, let $T_i$ be the set of all items for which $i$
is a top competitor.  Define $v'_i(S)=v_i(S-T_i)$ for any bundle~$S$. 
\item Let $(S'_1,\ldots,S'_n)$ and \hlhf{ $\{b_{i}(j)\}_{i \in N, j \in M}$} be the allocation and prices
returned by \\
$\TOPSTEAL((v'_1,\ldots,v'_n),(S_1, \cdots, S_n), M,t-1)$.
\label{step:no-top-comp-t>2-firstrecurse}
\item If $(S'_1,\ldots,S'_n)$ and the prices \hlhf{$\{b_{i}(j)\}_{i \in N, j \in
M}$} are an equilibrium with respect to $v_1,\ldots, v_n$ then return this allocation and prices.
\item Otherwise, find some bidder $i$ that can steal some item $j$ from bidder
$i'$, \hlhf{such that} $i$ is a top competitor for $j$.
\item Define $(S'_1,\ldots,S'_n)$ to be the allocation after $i$ steals $j$ from $i'$; thus $S'_k = S_k$ for $k \neq i,i'$, $S'_i = S_i \cup \{j\}$, $S'_{i'} = S_{i'} \setminus \{j\}$.
\label{step:no-top-comp-t>2-steal}
\item Return: \\ $\TOPSTEAL((v_1,\ldots,v_n),(S'_1,\ldots,S'_n),M,t)$
\label{step:no-top-comp-t>2-recurse}
\end{enumerate}
\end{enumerate}
We now analyze the running time of $\TOPSTEAL$  (see the appendix for the omitted proof). The base case will be proved first followed by the key claim of this section.

\begin{claim}\label{claim-base-case}
When the set~$M$ has only one element, $\TOPSTEAL((v_1,\ldots,v_n), (S_1, \cdots, S_n), M,t)$ reaches a no-overbidding equilibrium after one steal (i.e., $f_1(t)\leq 1$).
\end{claim}
\begin{proof}
When the item $j$ reaches a top competitor $i  $ no other bidder $i'$ wants to steal it since $v_i(j)\geq v_{i'}(j)$ by the definition of a top competitor.
\end{proof}

\begin{claim}\label{claim-compose-top}
Let $j$ be an item and let $i$ be a top competitor for this item. Let
$v'_i(S)=v_i(S|j)$. Let $S'=(S'_1,\ldots,S'_n)$ be a no-overbidding equilibrium
with respect to the valuations $(v'_i,v_{-i})$ where bidder $i'$ bids for item~$j'$ at $b'_{i'}(j')$. Then $S=(S'_i\cup \{j\},S'_{-i})$ is an equilibrium
with no overbidding with respect to $v_1,\ldots, v_n$ when for item $j'\neq j$
each bidder $i'$ bids \hlhf{$b_{i'}(j')=b'_{i'}(j')$}, and for item $j$ bidder $i$ bids $b_i(j)=v_i(j)$ and every other bidder $i'$ bids $b_i(j)=0$.
\end{claim}

\begin{proof}
We first show that there is no overbidding: this is certainly true for every
bidder $i'\neq i$ since \hlhf{the allocation and bidding of $i'$} 
did not change from the first equilibrium \hlhf{$S'$} in which there was no overbidding. In addition, bidder $i$
does not overbid \hlhf{either}. This is certainly true for every bundle $S$ where $j\notin S$ (since $v'_i(S)=v_i(S)$). For $j\in S$ we have that: 
\begin{align*}
v_i(S) &= v_i(\{j\})+v_i(S\setminus \{j\}|j)=v_i(\{j\})+v'_i(S\setminus\{j\}) \\
& \geq b_i(j)+\Sigma_{j'\neq j,j'\in S}b'_i(j') = \Sigma_{j'\in S}b_i(j').
\end{align*}
To see that $S$ is an equilibrium, observe that for every item $j'\notin S'_i
\cup \{j\}$ it holds that $v_i(j'|S'_i \cup \{j\})=v'_i(j'|S'_i)$, so bidder~$i$
will not steal any item from any other bidder $i'$. Similarly, since the
valuation of every other bidder $i'\neq i$ did not change and the prices of all
items except $j$ are the same, the only possible steal may involve some bidder
$i'$ stealing item $j$ from bidder $i$. However, $i$ is a top competitor for $j$
so $b_i(j)=v_i(\{j\})\geq v_{i'}(\{j\})\geq v_{i'}(j|S'_{i'})$, hence this steal is not profitable for $i'$. 
\end{proof}

\begin{lemma}\label{lemma-t=2}\footnote{One could in principle incorporate the next two lemmas into one by starting the induction from $t=1$. We explicitly prove the case of $t=2$ since it provides an explicit simple proof for the important setting of two bidders.}
When each item is competed by at most two bidders, i.e., $t = 2$, the procedure
\\ $\TOPSTEAL((v_1,\ldots, v_n), (S_1, \cdots, S_n), M,2)$ returns a no-overbidding equilibrium after $f_m(2)\leq m$ steals.
\end{lemma}
\begin{proof}
We prove this by induction on the number of items $m$, where the case of $m=1$ is proved in Claim \ref{claim-base-case}. We assume correctness for $m-1$ and prove for $m$. We divide into two cases:

\vspace{0.05in} \noindent \textbf{Case 1:} \emph{There \hlhf{exist an} item $j$ and \hlhf{a} bidder $i$ such that
\hlhf{$j\in S_i$} and $i$ is \hlhf{a} top competitor for $j$:} the fact that this is a no-overbidding equilibrium
follows since by the induction hypothesis \hlhf{$\TOPSTEAL((v_{i,
|\{j\}},v_{-i}), (S_i \setminus \{j\}, S_{-i}), M\setminus \{j\},2)$} returns a
no-overbidding equilibrium with respect to \hlhf{$(v_{i, |\{j\}}, v_{-i})$}
and then \hlhf{we can apply} Claim~\ref{claim-compose-top}. Notice that by the
induction hypothesis, \hlhf{$\TOPSTEAL((v_{i, |\{j\}},v_{-i}), (S_i \setminus
\{j\}, S_{-i}), M\setminus\{j\},2)$} makes at most $m-1<m$ steals.

\vspace{0.05in} \noindent \textbf{Case 2:} \emph{There is no item $j$ and bidder $i$ such that \hlhf{$j\in S_i$} and
$i$ is \hlhf{a} top competitor for $j$:} if we reached an equilibrium in the
first step we are done. Else, some bidder $i$ steals item~$j$
from another bidder~$i'$. Since $|C_j|=2$ and $i'\in C_j$, by assumption we have that $i$ is a top competitor of~$j$. The recursive call to
$\TOPSTEAL$ in Step 2(e)
will satisfy Case 1, so by the first part of this proof the recursive call 
will reach a no-overbidding equilibrium after at most $m-1$ steals. Combining
this with the one steal in 
Step~2(d)
we have at most $m$ steals in total, as claimed.
\end{proof}

\begin{lemma}\label{lemma-topsteal-last} For general values of~$t$, the
procedure
$\TOPSTEAL((v_1,\ldots, v_n), (S_1, \cdots, S_n),M,t)$ returns a no-overbidding equilibrium after $f_m(t)\leq 1+f_{m}(t-1)+f_{m-1}(t)$ steals.
\end{lemma}

\begin{proof}
We again prove using induction on the number of items and on the size of $t$ where the base case $m=1$ was proven in Claim \ref{claim-base-case} and the case of arbitrary $m$ and $t=2$ was proven in Lemma \ref{lemma-t=2}. We divide again into two cases:

\vspace{0.1in} \noindent \textbf{Case 1.} \emph{There \hlhf{exist an} item $j$ and \hlhf{a} bidder $i$ such that
\hlhf{$j\in S_i$} and $i$ is \hlhf{a} top competitor for $j$:} the fact that this is a no-overbidding equilibrium
follows since by the induction hypothesis, \hlhf{$\TOPSTEAL((v_{i,
|\{j\}},v_{-i}), (S_i \setminus \{j\}, S_{-i}),  M\setminus \{j\},t)$} returns a
no-overbidding equilibrium with respect to \hlhf{$(v_{i, |\{j\}},v_{-i})$}
and then \hlhf{we can apply} Claim~\ref{claim-compose-top}. Notice that by the
induction hypothesis, \hlhf{$\TOPSTEAL((v_{i, |\{j\}},v_{-i}), (S_i \setminus
\{j\}, S_{-i}),  M\setminus \{j\},t)$} makes at most $f_{m-1}(t)$ steals.

\vspace{0.1in} \noindent \textbf{Case 2.} \emph{There is no item $j$ and bidder $i$ such that \hlhf{$j\in S_i$} and
$i$ is \hlhf{a} top competitor for $j$:} first, observe that if $v_1,\ldots,
v_n$ is a $t$-restricted instance, then $(v'_1, \ldots, v'_n)$ is a $(t-1)$-restricted
instance. If running $\TOPSTEAL$ on this $(t-1)$-restricted instance results
in an equilibrium \hlhf{S'}, the lemma is proved and we made only  at most $f_m(t-1)$ steals.

Otherwise, $S'$ is not an equilibrium since some bidder $i$ can
steal some item $j$ from bidder $i'$, and $i$ is a top competitor for $j$. This
is true since for every bidder $i$ and item $j$ we have 
$v_i(j|S'_i)=v'_i(j|S'_i)$ unless $i$ is a top competitor for $j$. 
The recursive call to
$\TOPSTEAL$ in Step~2(f)
will satisfy Case 1, so by the first part of this proof the recursive call 
will reach a no-overbidding equilibrium after at most $f_{m-1}(t)$ steals. 
Combining this with the $f_m(t-1)$ steals in 
Step~2(e)
and the one steal in 
Step~2(b)
we have at most $1+f_m(t-1)+f_{m-1}(t)$ steals in total, as claimed.
\end{proof}

\section{An Impossibility Result for Subadditive Valuations}
\label{sec:subadd}

Unlike submodular valuations, with subadditive valuations we have no hope to find an equilibrium.

\begin{theorem}
Exponential communication is needed to determine whether there is a no-overbidding equilibrium in combinatorial auction with subadditive bidders. Furthermore, even when an equilibrium is guaranteed to exist, exponential communication is needed to find it. The results hold even if there are only two players and even if the value of every bundle is in $\{0,1,2\}$.
\end{theorem}

We prove only the first statement, \hlhf{and} the second one is an easy
corollary of the first. We prove the theorem for two players; extending the
theorem for more than two players is straightforward by adding bidders with
valuations that are identically zero. For the proof we use a set-pair system:
\begin{definition}
A set-pair system $\mathcal S=\{(S^r_1,S^r_2)\}_r$ is \emph{good} if the following holds:
\begin{enumerate}
\setlength{\itemsep}{0pt}
\item For every $(S^r_1,S^r_2)\in \mathcal S$ it holds that $|S^r_1|=|S^r_2|=\frac m 4$ and $S^r_1\cap S^r_2=\emptyset$.
\item For every $r\neq l$ we have that $0 < |S^r_1\cap S^l_2|\leq \frac m 8$.
\end{enumerate}
\end{definition}

It is well known that good set-pair systems with exponentially many
allocations exist (e.g., \cite{N02}).  \hlhf{Given a good set-pair system, we reduce the
disjointness problem to the problem of finding a no-overbidding equilibrium in
the simultaneous second-price auction involving valuations to be defined below.
The disjointness problem involves two players, each player $i$ having} a string
$A^i=\{0,1\}^t$, and the problem is to determine whether there is some $k$ such
that $A^1_k=A^2_k=1$. It is known that deciding disjointness (or finding $k$ if
we are given that such $k$ exists) requires $\Omega(t)$ bits of communication \cite{KN97}.
\hlhf{In our reduction, the size~$t$ will be equal to the size of the good set
system, which is exponential in~$m$.}  We construct the following valuation
functions ($i=1,2$) \hlhf{given a bit vector $A^i \in \{0, 1 \}^t$.  As assumed by normalization, $v_i(\emptyset) = 0$; for $S \neq \emptyset$,}
\begin{align*}
v_i(S) = 
\begin{cases}
\setlength{\itemsep}{0pt}
2, & \mbox{if }|S|\geq \frac{3m} 4 +1; \\
2, & \mbox{if $\exists r$ such that $S^r_i \subseteq S$ and $A^i_r=1$};\\
1, & \mbox{otherwise.} 
\end{cases} 
\end{align*}
\hlhf{It is straightforward to verify that these value functions are
subadditive.} The theorem is a corollary of the following two claims:

\begin{claim}\label{claim-subadditive-exists-eq}
If there exists $k$ such that $A^1_k=A^2_k=1$ then there is a no-overbidding equilibrium.
\end{claim}

\begin{proof}
Consider the following strategies: each bidder $i$ bids some small $\epsilon>0$
for every $j\in S_i^k$ and $0$ otherwise. Denote the outcome \hlhf{allocation in
the simultaneous second-price auction} of these two strategies by $(S_1,S_2)$.
We claim that these strategies form an equilibrium (observe that there is no
overbidding in these strategies). First, observe that no item gets positive bids
from both bidders simultaneously (by the properties of a good set-pair system) and that each positively bids on all items of some bundle that has a value of $2$ for him. Thus, for every bidder $i$, $v_i(S_i)=2$. Since $v_i(M)=2$ as well and since the total payment of each bidder is $0$, there is no strategy that gives bidder $i$ a profit larger than $2$ (regardless of the strategy of the other bidder).
\end{proof}

\begin{claim}
If there is no $k$ such that $A^1_k=A^2_k=1$ then there is no equilibrium with no overbidding.
\end{claim}

\begin{proof}
Suppose that there is some \hlhf{no-overbidding} equilibrium, and denote by
$b_i(j)$ the bid of bidder $i$ for item $j$ in that equilibrium. Let the
allocation in \hlhf{the} equilibrium be $(S_1,S_2)$. Observe that by the
properties of a good set-pair system and the definition of the valuations it either holds that $v_1(S_1)\neq 2$ or that $v_2(S_2)\neq 2$. Thus, we may assume without loss of generality that $v_1(S_1)\leq 1$.





Let us call a set of items $U$ \emph{unprotected} if $v_1(U)=2$ 
and $\Sigma_{j \in U} b_2(j) < 1$. Note that if $U$ is unprotected and bidder
1 bids $b_2(j)+\epsilon$ for every $j \in U$ and 0 for every $j \not\in U$,
then for sufficiently small positive $\epsilon$ this satisfies
no-overbidding%
 (as the sum of all bids is less than 1) and strictly
increases bidder 1's profit. Thus, existence of an unprotected set
contradicts our hypothesis that $(b_1,b_2)$ is an equilibrium.

Let $T$, $|T|=m/4$, be such that $v_1(T)=2$. By the properties of a
good set-pair system it holds that $v_2(T)=1$. If $\Sigma_{j \in T} b_2(j)<1$
then $T$ is an unprotected set, so suppose henceforth 
that $\Sigma_{j \in T} b_2(j)=1$. For every $j \notin T$ we have
$v_2(T+\{j\})=1$ by the properties of a good set system,
since $v_1(T)=2$ and $T + \{j\}$ has only one element that
does not belong to $T$. Now the assumption that bidder 2
does not overbid implies $b_2(j)=0$ for all $j \notin T$.
Consequently $\Sigma_{j \in M} b_2(j) = \Sigma_{j \in T} b_2(j) = 1$.
Now choose any $j'$ such that $b_2(j')>0$ and observe that
$M - \{j'\}$ is an unprotected set.
\end{proof}

\section{An Impossibility Result for XOS Valuations}
\label{sec:xos}

This section proves an impossibility result for XOS algorithms: we
essentially show that every efficient algorithm for XOS valuations 
must use techniques very different \hlhf{from} traditional algorithms for XOS
valuations. Section \ref{subsec-xos-impossibility} proves the impossibility
result itself, but we start by negatively answering a related open
question from \cite{CKS10}: does the XOS algorithm of \cite{CKS10} \hlhf{end} in polynomial time if the valuations are submodular? 

\vspace{0.1in}\noindent\textbf{Exponential Convergence for the XOS Algorithm with Submodular Bidders.} We show that if the bidders have submodular valuations, then the algorithm
of \cite{CKS10} may terminate after exponentially many steps, even when there
are only two bidders. The key for this example is the XOS oracle: we
choose the prices in a legitimate way according to the algorithm, but one
that will lead to a long convergence path. The dependence of this lower
bound on the implementation of the XOS oracle is unavoidable; 
as discussed earlier,
the stealing procedure presented in Section~\ref{sec:submodular} is in fact a special case of the algorithm of \cite{CKS10} that essentially differs only in the implementation of the XOS oracle. The example itself appears in Appendix \ref{app-xos-example}.

\subsection{The Impossibility Result}\label{subsec-xos-impossibility}
Recall that an XOS oracle for an XOS valuation $v$ returns
for each bundle $S$ 
a maximizing clause for $S$ in $v$. We refer to the coefficients
of this maximizing clause as ``prices.''
In this section we denote an XOS oracle for a valuation $v$ by $\mathcal O_{v}$. We now prove a lower bound on deterministic algorithms that find a \emph{traditional} equilibrium. For simplicity, we restrict ourselves to the case where there are only two players (the impossibility result trivially extends to any number of players by adding players with valuations that are identically zero).

\begin{definition}
A no overbidding equilibrium $(S_1, S_2)$ is called \emph{traditional} with
respect to some XOS oracles $\mathcal O_{v_1}$ and $\mathcal O_{v_2}$ if, for
each bidder $i$ and item $j$, if $j\in S_i$ then $b_i(j)$ equals the price of $j$ in $\mathcal O_{v_i}(S_i)$, and if $j\notin S_i$ then $b_i(j)=0$.
\end{definition}

It is easy to verify that the equilibrium obtained by the algorithm of \cite{CKS10} is traditional.

\begin{theorem}
Let $A$ be a deterministic algorithm that always produces a traditional equilibrium with respect to some XOS oracles $\mathcal O_{v_i}$. Suppose that $A$ is only allowed to make demand and value queries, as well as XOS queries to the oracles $\mathcal O_{v_i}$. Then, $A$ makes an exponential number of queries in the worst case. The theorem holds even if there are only two bidders with identical valuations.
\end{theorem}

The proof consists of three parts. First, we present a family of XOS valuations (sensitive valuations) and give conditions for equilibrium with this family. We then show that demand queries to sensitive valuations are almost useless: every demand query can be simulated by polynomially many value queries. Finally, we prove that any algorithm for finding equilibrium when bidders have sensitive valuations that uses value queries makes exponentially many value queries in the worst case. Our bound will then follow.

\subsubsection{Proof part I: conditions for equilibrium}

Let $m>20$ be an odd integer and let $m'=\lfloor m/2 \rfloor$. Given
non-negative $k_S$'s such that for every bundle~$S$, $|S|=m'+1$, we have $\frac 1 4>k_S >0$, consider the valuation $v$ that is the maximum of the following additive valuations (and hence XOS by definition):
\begin{enumerate}
\setlength{\itemsep}{0pt}
\item There is an additive valuation $C_j$ for every item $j$ that gives a value of $m'-20$ for $j$ and value of $0$ for every other item.
\item For every bundle $S$, $|S|=m'$, there is an additive valuation $A_S$ that gives a value of $1$ to every item $j\in S$ and value $0$ for any other item.
\item For every bundle $S$, $|S|= m'+1$, there is an additive valuation
$M_{S,j}$, for one item $j\in S$. $M_{S,j}$ gives a value of $\tfrac 1 4+k_S$
for item $j$, a value of $1$ for every item $j'\in S \setminus \{j\}$, and $0$ for every item that is not in $S$.
\item For every bundle $S$, $|S|= m'+10$, there is an additive valuation $B_S$ that gives a value of $\frac {m'+1} {m'+10}$ for every item in $S$ and $0$ otherwise.
\end{enumerate}

We call such valuations \emph{sensitive}. An XOS oracle for a sensitive
valuation is \emph{standard} if for every bundle $S$ it returns an additive
valuation that is specified in the above definition. It will sometimes be easier
to work with an explicit description of a sensitive valuation (see proof in the appendix):

\begin{proposition}\label{prop-xos-equiv}
Let $v'$ be a sensitive valuation defined by $k_S$'s. Let $v_i$ be the following valuation:
$$
v_i(S)=
\begin{cases}
\setlength{\itemsep}{0pt}
m'-20 &  |S|\leq m'-20,\\
|S| &   m'-20 < |S| \leq m',\\
m'+\frac 1 4+\max_{S' \subset S} k_{S'} &  m' < |S| < m'+10,\\
m'+1 & |S| \geq m'+10.\\
\end{cases}
$$
For every $S$ we have that $v(S)=v'(S)$.
\end{proposition}




\begin{definition}
A bundle $S$ of size $m'+1$ \hlhf{and a corresponding XOS clause $M_{S,j}$} are called a
\emph{$j$-local maximum} of a valuation $v$ with respect to \hlhf{an} item $j\in
S$ if $v(S)\geq v(M-S+\{j\})$.
\end{definition}

The definition of $j$-local maximum and of $M_{S,j}$ implies that the XOS
clause of~$S$ puts a weight strictly smaller than~$1$ on~$j$.  The condition
$v(S) < v(M - S + \{j\})$ alone does not imply that $S' = M - S + \{j\}$ is a
$j$-local maximum, because even though $v(S') > v(M - S' + \{j\}) = v(S)$, the
XOS clause for~$S'$ may be $M_{S', j'}$ for an item $j' \neq j$.  The reference
to the XOS clause is important for traditional equilibria, in which a bidder
allocated a $j$-local maximum necessarily bids strictly less than~$1$ on~$j$. The proof of the next proposition is in the appendix.

\begin{proposition}
\label{prop:eq-char}
Let $v$ be a sensitive valuation and let $(S_1, S_2)$ be a traditional equilibrium of two bidders with the same valuation $v$ with respect to the standard XOS representation. Then, either $|S_1|=m'$, $|S_2|=m-m'=m'+1$, and $S_2$ is a $j$-local maximum of $v$ (for some item $j\in S_2$), or $|S_2|=m'$, $|S_1|=m-m'$, and $S_1$ is a $j$-local maximum of $v$ (for some item $j\in S_1$).
\end{proposition}

\subsubsection{Proof part II: simulating demand queries by value queries}

We now show how to simulate a demand query using value queries, by adapting techniques from~\cite{BDO12}. Since we show an impossibility result and since the values of all bundles of sizes other than $m'+1$ are known, we may consider only demand queries that may return a bundle of size $m'+1$. 

\begin{definition}
Let $\dq$ be a demand query. A bundle $S$, $|S|=m'+1$, is \emph{covered} by
$\dq$ if there is a sensitive valuation $v$ such that when querying $v$
for $\dq$ a profit-maximizing bundle is $S$. 
\end{definition}

The main point here is that the number of bundles covered by a query is bounded
(see proof in the appendix):  

\begin{lemma}\label{lemma-bounded-cover}
Fix a demand query $\dq$. Let $\mathcal S=\{S|\hbox{$S$ is covered by
$\dq$}\}$. Then, $|\mathcal S|\leq \poly(m)$. 
\end{lemma}

\begin{corollary}
Any demand query $\dq$ for a sensitive valuation $v$ can be simulated by polynomially many value queries.
\end{corollary}
\begin{proof}
Query all the polynomially many sets covered by~$\dq$. Among these sets and all the others (whose values are already known), return a profit-maximizing one.
  \end{proof}

\subsubsection{Proof part III: the power of value queries}

We now show that exponentially many value queries are needed to find a $j$-local maximum of a sensitive valuation $v$. The adversary will construct the hard valuation $v$ by following the algorithm: whenever the algorithm queries some bundle $S$, $|S|=m'+1$ the adversary will make sure that $S$ is not a local maximum. Obviously, since $v$ is sensitive there is no use for the algorithm to query bundles of size different from $m'+1$, since the values of those bundles are known in advance. We will assume that whenever the algorithm makes a value query to some bundle $S$, it also makes an XOS query to the same bundle (this only makes our bound stronger). Specifically, we show that:

\begin{lemma}\label{lemma-xos-part3}
Any deterministic algorithm for finding a $j$-local maximum in sensitive
valuations makes at least \hlhf{$\frac {2^{0.75 m'-1}} {m'}$} queries in the worst case.
\end{lemma}

In the proof we use a graph $G$ that is composed of ${m \choose {m'+1}}$
vertices, each \hlhf{of which} is associated with a different bundle of size
$m'+1$. Two vertices $S$ and $S'$ of $G$ \hlhf{are connected by} an edge if and
only if $S=M-S'+\{j\}$, for some $j\in S'$. This graph is known as the \emph{odd graph} $O_{m'+1}$.  \hlhf{We
will refer to the subsets of~$M$ (of size~$m' + 1$) and vertices of~$O_{m'+1}$
interchangeably, and the meaning should be clear, for example, when we say a set
$S$ is a neighbor of another set~$S'$.}  We will need the following isoperimetric inequality on odd graphs (this inductive proof is inspired by \cite{H76}):

\begin{proposition}{\rm \sc (Isoperimetric inequalities for odd graphs.)}
Let $O_n=(V,E)$ be an odd graph. Then, for any $S\subseteq V$, $k=|S|$, at least
$(n- \tfrac 4 3 \log k)k$ edges have exactly one end in $S$. As a corollary,
since $O_n$ is $n$-regular, the number of neighbors of $S$ is at least $
\left(n - \tfrac 4 3 \log k \right)\frac{k}{n}$. 
\end{proposition} 

\begin{proof}
We use the following binary representation of the vertices of $O_n$: the $t$-th index of $v\in V$ is $1$ if and only if $v$ corresponds to a set that contains the $t$-th element. Let $\EE(k)$ denote the maximum possible number of edges inside a set $S\subseteq V$ of size~$k$. By the regularity of~$O_n$, the number of edges
having exactly one end in~$S$ is at least $nk - 2\EE(k)$.  It therefore suffices
to show $\EE(k) \leq \frac {2 k \log_2 k}3$, and we prove this by induction on~$k$.  For the base case of $k=2$ observe that $\EE(2)=1$.  

We now assume correctness for $k-1$ and prove for $k$. Consider a set $S$ of
size $k$. We claim that there exist two distinct
indices $x_1,x_2$ such that for at most
two-thirds of the elements $v\in S$ it holds 
that $v_{x_1} \neq v_{x_2}$.
To see this, consider some $v$ and two randomly chosen 
distinct indices $x'_1$ and $x'_2$.  
The probability that $x'_1=1$ is exactly $\frac {n} {2n-1}$
since the Hamming weight of the representation of every
vertex is exactly $n$. The conditional probability that $x'_2=0$ given 
$x'_1=1$ is $\frac{n-1}{2n-2} = \frac 1 2$.
The probability of the event $x'_1=0, \, x'_2=1$ 
is the same, by symmetry, and therefore
the probability that $v_{x'_1} \neq v_{x'_2}$
is $2\cdot \frac {n} {2n-1} \cdot \frac 1 2 \leq \frac 2 3$. 
By linearity of expectation, for randomly chosen $x'_1$ and $x'_2$, the expected
number of vertices $v \in S$ such that $v_{x'_1} \neq v_{x'_2}$ is at
most $\frac {2|S|} 3$. Hence, there exist two indices $x_1,x_2$
such that at most two-thirds of the elements in~$S$ have 
exactly one~$1$ at the two coordinates.  
Call this set of elements $S_0$ and let $S_1=S-S_0$. Let
$\ell=|S_0| \leq \tfrac {2 k} 3$.

The edges in $S$ are those inside $S_0$, those inside $S_1$, and those
between $S_0$ and $S_1$. Crucially, there are at most
$\ell$ edges between $S_0$ and $S_1$: consider some $v\in S_0$.  The number
of neighbors of $v$ in $S$ is at most one since $v$ has no neighbor whose
coordinates at $x_1$ and~$x_2$ are both~$0$, and at most one neighbor
whose coordinates at $x_1$ and~$x_2$ are both~$1$.
We have:
\begin{align*}
\EE(k) &\leq \max_{1 \leq \ell\leq \frac {2 k} 3} (\ell+\EE(\ell)+\EE(k-\ell))
\\
       & \leq \max_{1 \leq \ell \leq \frac {2 k} 3} \left[\ell + \frac {2 \ell\log \ell}{3}  + \frac {2 (k-\ell)\log(k-\ell)}3 \right]\\
      &\leq \frac 2 3 \left( k + \max_{1 \leq \ell \leq \frac{2k}{3}} \left[
\ell \log \ell + (k-\ell) \log(k-\ell) \right] \right) \\
       & \leq \frac 2 3 \left( k + k \log \left( \frac{k}{2} \right) \right) = \frac {2 k \log k}{3}.
\end{align*}
\end{proof}


\subsection*{Proof of Lemma \ref{lemma-xos-part3}}
Denote \hlhf{by $Q_i$} the $i$-th query that the algorithm makes (to a bundle of
size $m'+1$). For the first query $Q_1$ we will \hlhf{return} a value of $m'+\frac 1 4 + \epsilon$ and fix 
the appropriate \hlhf{clause} \hlhf{$M_{Q_1,j}$}, for some $j$. At this point the
algorithm cannot be sure if $Q_1$ is a local maximum, since the value of
$M-S+\{j\}$ is still undetermined. In general, when the algorithm makes the
$l$-th query $Q_l$, \hlhf{tentatively} we would like to set  $v(Q_l)=m'+\frac 1 4 + l\cdot
\epsilon$, and set the XOS clause of $Q_l$ to $M_{Q_l,j}$, where $j$ is such that $Q_l$ was not \hlhf{queried} yet.

The problem with this is that all \hlhf{neighbors} of $Q_l$ may have been queried, so a local maximum may quickly be found by the algorithm. To fix this, 
whenever the algorithm queries $Q_l$, before determining $v(Q_l)$ and the
maximizing clause, \hlhf{we first consider ``small'' patches of the graph that
will be cut off by the removal of~$Q_l$; we will first fix the values for
vertices in these patches: in each small connected component that is cut off,
the values assigned is increasing from vertices furthest from~$Q_l$
to those closest.  At last, we give the largest value to $Q_l$.  We will mark
these vertices as ``colored'', and remove them from the graph, together with
$Q_l$, for future rounds.  Importantly, the XOS clause $M_{Q_l, j}$ for $Q_l$
is chosen such that $j$ is in an uncolored neighbor of~$Q_l$, and therefore the
algorithm cannot know whether it has found a local maximum or not.}

More formally, the process \hlhf{is} as follows. Let $\epsilon>0$ be such that
\hlhf{$2^m \cdot \epsilon<\frac 1 4$}. 

\begin{enumerate}
\setlength{\itemsep}{0pt}
\item Set $x = 1$. Let $\CC=\QQ=\emptyset$.
\item For each query $Q_i$:
	\begin{enumerate}
	\setlength{\itemsep}{0pt}
	\item If $Q_i \in \CC$, return the already fixed value and clause, and
go for the next query.  Otherwise let $\QQ=\QQ+\{Q_i\}$ and proceed.
	\item For each connected component $CC$ of $G-\CC-\QQ$ of size less than
$c = 2^{0.75 m'- 1}$:
	\begin{enumerate}
		\item \label{step:dist} Let $D$ be the diameter of~$CC$, let $x
= x + D + 1$.  For each $S \in CC$ that is of distance $d$ from
$Q_i$ let $v(S)=m'+\frac 1 4+ (x-d)\cdot \epsilon$.  Let $(S, S', \cdots,
Q_i)$ be a shortest path from $S$ to~$Q_i$ in~$CC$.  Let $S \cap S'$ be $\{j\}$.
Let the maximizing clause of~$S$ be $M_{S, j}$.
		\item Add $CC$ to $\CC$.
	\end{enumerate}
	\item \label{step:Qi} Update $x = x + 1$.  Let $v(Q_i)=m'+\frac 1 4+ x \cdot \epsilon$.  If
there is an uncolored neighbor $S'$ of~$Q_i$, let $S' \cap Q_i$ be $\{j\}$ and let the maximizing clause of~$Q_i$ be $M_{Q_i, j}$.  Otherwise the process
terminates.
	\end{enumerate}
\end{enumerate}

Since the distance~$d$ in step~\ref{step:dist} is bounded by the diameter
of the connected component, we are guaranteed that the value we assign during
the process is always increasing.  Therefore until we cannot find an
uncolored neighbor in step~\eqref{step:Qi}, the algorithm does not
find a $j$-local maximum and therefore by Proposition~\ref{prop:eq-char} it does
not find a traditional equilibrium.  At the moment we cannot do this, the
connected components newly colored in that step together with~$Q_i$ forms one
connected component~$CC^*$ that was not colored in previous steps, and therefore
is of size at least $c$.  The crucial point is that $\QQ$ contains the set of
neighbors of~$CC^*$, and by the isoperimetric inequality, $|\QQ| \geq |CC^*|
\cdot (1 - \tfrac{4 \log |CC^*|}{3m'}) =  \tfrac{2^{0.75 m' - 1}}{m'}$, as needed.

\subsubsection*{Acknowledgments}
We thank Noam Nisan for pointing our attention to \cite{H76}.

\bibliographystyle{plain}
\bibliography{bib}

\begin{thebibliography}{10}

\bibitem{AM04}
Nir Andelman and Yishay Mansour.
\newblock Auctions with budget constraints.
\newblock In {\em SWAT}. 2004.

\bibitem{BDO12}
Ashwinkumar Badanidiyuru, Shahar Dobzinski, and Sigal Oren.
\newblock Optimization with demand oracles.
\newblock In {\em Proceedings of the 13th ACM Conference on Electronic
  Commerce}. ACM, 2012.

\bibitem{BhawalkarR11}
Kshipra Bhawalkar and Tim Roughgarden.
\newblock Welfare guarantees for combinatorial auctions with item bidding.
\newblock In {\em SODA}, 2011.

\bibitem{BlumrosenNisan07}
Liad Blumrosen and Noam Nisan.
\newblock Combinatorial auctions (a survey).
\newblock In Noam Nisan, Tim Roughgarden, Eva Tardos, and Vijay Vazirani,
  editors, {\em Algorithmic Game Theory}. Cambridge University Press, 2007.

\bibitem{CP14}
Yang Cai and Christos Papadimitriou.
\newblock Simultaneous bayesian auctions and computational complexity.
\newblock In {\em Proceedings of the fifteenth ACM conference on Economics and
  computation}, pages 895--910. ACM, 2014.

\bibitem{CG10}
Deeparnab Chakrabarty and Gagan Goel.
\newblock On the approximability of budgeted allocations and improved lower
  bounds for submodular welfare maximization and gap.
\newblock {\em SIAM Journal on Computing}, 39(6):2189--2211, 2010.

\bibitem{CKS10}
George Christodoulou, Annam\'aria Kov{\'a}cs, and Michael Schapira.
\newblock {Bayesian combinatorial auctions}.
\newblock {\em Automata, Languages and Programming}, pages 820--832,
  \hlhf{2008}.

\bibitem{DNO14}
Shahar Dobzinski, Noam Nisan, and Sigal Oren.
\newblock Economic efficiency requires interaction.
\newblock In {\em ACM STOC}. 2014.

\bibitem{DNS05}
Shahar Dobzinski, Noam Nisan, and Michael Schapira.
\newblock Approximation algorithms for combinatorial auctions with
  complement-free bidders.
\newblock In {\em ACM STOC}, 2005.

\bibitem{DS06}
Shahar Dobzinski and Michael Schapira.
\newblock An improved approximation algorithm for combinatorial auctions with
  submodular bidders.
\newblock In {\em Proceedings of the Seventeenth Annual {ACM-SIAM} Symposium on
  Discrete Algorithms, {SODA} 2006, Miami, Florida, USA, January 22-26, 2006},
  pages 1064--1073, 2006.

\bibitem{Feige09}
Uriel Feige.
\newblock On maximizing welfare when utility functions are subadditive.
\newblock {\em SIAM J. Comput.}, 39(1):122--142, 2009.

\bibitem{FFGL13}
Michal Feldman, Hu~Fu, Nick Gravin, and Brendan Lucier.
\newblock Simultaneous auctions are (almost) efficient.
\newblock In {\em Proceedings of the 45th annual ACM symposium on Symposium on
  theory of computing}, pages 201--210. ACM, 2013.

\bibitem{FKL12}
Hu~Fu, Robert Kleinberg, and Ron Lavi.
\newblock Conditional equilibrium outcomes via ascending price processes with
  applications to combinatorial auctions with item bidding.
\newblock In {\em ACM Conference on Electronic Commerce}, 2012.

\bibitem{H76}
Sergiu Hart.
\newblock A note on the edges of the n-cube.
\newblock {\em Discrete Mathematics}, 14(2):157--163, 1976.

\bibitem{HassidimKMN11}
Avinatan Hassidim, Haim Kaplan, Yishay Mansour, and Noam Nisan.
\newblock Non-price equilibria in markets of discrete goods.
\newblock In {\em ACM Conference on Electronic Commerce}, pages 295--296, 2011.

\bibitem{KN97}
Eyal Kushilevitz and Noam Nisan.
\newblock {\em Communication Complexity}.
\newblock Cambridge University Press, New York, NY, USA, 1997.

\bibitem{LLN06}
Benny Lehmann, Daniel Lehmann, and Noam Nisan.
\newblock {Combinatorial auctions with decreasing marginal utilities}.
\newblock {\em Games and Economic Behavior}, 55(2):270--296, 2006.

\bibitem{PLST12}
Renato~Paes Leme, Vasilis Syrgkanis, and Eva Tardos.
\newblock Sequential auctions and externalities.
\newblock In {\em SODA}, pages 869--886, 2012.

\bibitem{LB10}
Brendan Lucier and Allan Borodin.
\newblock Price of anarchy for greedy auctions.
\newblock In {\em SODA}, 2010.

\bibitem{N02}
Noam Nisan.
\newblock The communication complexity of approximate set packing and covering.
\newblock In {\em ICALP}, pages 868--875, 2002.

\bibitem{NisanBlog}
Noam Nisan.
\newblock The computational complexity of pure nash.
\newblock 2009.
\newblock Blog post in Turing's Invisible Hand. November 19.

\bibitem{R12}
Tim Roughgarden.
\newblock The price of anarchy in games of incomplete information.
\newblock In {\em ACM Conference on Electronic Commerce}, pages 862--879, 2012.

\bibitem{SW95}
Carla~D. Savage and Peter Winkler.
\newblock Monotone {G}ray codes and the middle levels problem.
\newblock {\em Journal of Combinatorial Theory, Series A}, 70(2):230--248,
  1995.

\bibitem{ST12}
Vasilis Syrgkanis and Eva Tardos.
\newblock Bayesian sequential auctions.
\newblock In {\em ACM EC}, 2012.

\bibitem{ST13}
Vasilis Syrgkanis and Eva Tardos.
\newblock Composable and efficient mechanisms.
\newblock In {\em STOC}, 2013.

\bibitem{V08}
Jan Vondr{\'a}k.
\newblock Optimal approximation for the submodular welfare problem in the value
  oracle model.
\newblock In {\em STOC}, 2008.

\end{thebibliography}

\appendix

\section{Missing Proofs of Section \ref{sec:submodular}}

\subsection{Pseudo-Polynomial Convergence Time}

\hlhf{In this subsection we assume that all valuations are rational numbers.} Let $v_{max}=\max_iv_i(M)$.

\begin{claim}
 Let $\Delta$ be such that $v_i(j|S)$ is a multiple of $\Delta$, for every bidder $i$, item $j$, and subset $S$. The number of steals that the iterative stealing procedure makes is at most $n\cdot \frac {v_{max}} {\Delta}$.
\end{claim}
\begin{proof}
By Claim \ref{claim-increase-welfare} after every step the welfare increases by
\hlhf{at least} $\Delta$. A trivial upper bound on the welfare is $n\cdot \frac {v_{max}} {\Delta}$. Thus the procedure ends after $n\cdot \frac {v_{max}} {\Delta}$ steals in every implementation.
\end{proof}

Next we show a specific implementation that gives a slightly better bound
than this last result. 

\begin{claim}\label{claim-pseudo-better}
For each bidder $i$ and item $j$, let \hlhf{$\Delta_j^i=|\{v_i(j|S)\}_{S\in
2^M}|$}. There exists an implementation of the iterative stealing procedure that terminates after at most $\Sigma_{i\in N}\Sigma_{j\in M}\Delta^i_j$.
\end{claim}
\begin{proof}
We start with some allocation $(S_1,\ldots ,S_n)$. Each bidder $i$ will maintain
some order $>_i$ on the items according to which Step \ref{step-update} is
implemented.  After each steal we possibly update $>_i$. We begin with some
order $>_i$ such that for each bidder $i$ it holds that each item $j\in S_i$
\hlhf{precedes} every item $j'\notin S_i$.  The internal order of items within
these two groups is arbitrary.  When bidder $i$ steals item $j$ from bidder $i'$
we change $>_i$ and $>_{i'}$ so that $j$ is the \hlhf{last} item in both $>_i$
\hlhf{(within the new $S_i$)} and in $>_{i'}$ \hlhf{(among all items)}. The order of the rest of the items remains the same.

\HuNote{We should consider taking this paragraph on particular ordering out, since the proofs of the different parts of the theorem reuse it.}

For every item $j$ let the price of item $j$ be $p_j=\max_ib_i(j)$. An easy observation is that using this implementation the price $p_j$ of an item can only go up. Consider bidder $i$ that holds item $j$. First, right after bidder $i$ steals item $j$ from some other bidder $p_j$ can only go up since the price is simply bidder $i$'s marginal value for $j$. If another item $j'$ was stolen from $i$ then the price does not change if $j>_i j'$ and cannot decrease by submodularity if $j>_i j'$. If bidder $i$ steals another item $j'$ then $j'$ is the smallest item in $>_i$ and the price of $j$ did not change as well.

Using the fact that the price of an item can only increase it is easy to get a bound on the number of steals. Recall that after bidder $i$ steals item $j$ the price $p_j$ increases to $v_i(j|S_i)$ and any other price $p_{j'}$ cannot decrease. Thus bidder $i$ may steal each item $j$ at most $\Delta_j^i$ times. Consequently, the total number of steals is at most $\Sigma_{i\in N}\Sigma_{j\in M}\Delta^i_j$.
\end{proof}

\subsection{Budget Additive Bidders}

Recall that a valuation $v$ is \hlhf{budget} additive if there exists some $b$
such that for every $S$ we have that $v(S)=\min(b,\Sigma_{j\in S}v(\{j\}))$.
\hlhf{This $b$ is known as the \emph{budget} for the valuation.  In the sequel we will denote by $b_i$ the budget of bidder~$i$.}

The procedure for \hlhf{budget} additive bidders is similar to the iterative stealing procedure
presented earlier. We start with some allocation $(S_1,\ldots ,S_n)$. Each
bidder $i$ will maintain some order $>_i$ on the items according to which Step
\ref{step-update} is implemented. After each steal we will possibly update
$>_i$. We begin with some order $>_i$ such that for each bidder $i$ it holds
that each item $j\in S_i$ \hlhf{precedes} every item $j'\notin S_i$.  The
internal order of items within these two groups is arbitrary. When bidder $i$
steals item $j$ from bidder $i'$ we change $>_i$ and $>_{i'}$ so that $j$ is the
last item in both $>_i$ \hlhf{(within the new $S_i$)} and $>_{i'}$ \hlhf{(among
all items)}. The order of the rest of the items remains the same. For every item $j$ let the price of item $j$ be $p_j=\max_ib_i(j)$. Similarly to the proof of Claim \ref{claim-pseudo-better}, using this implementation the price of an item can only go up.

Given some allocation $(S_1,\ldots, S_n)$ we call an item $j\in S_i$ 
\hlhf{\emph{loose} for bidder~$i$ if $p_j < v_i(\{j\})$;}
otherwise the item is called \emph{tight}.  A loose item~$j$ is \emph{strongly loose} if $p_j = 0$, otherwise it is \emph{weakly loose}.  We need a few easy claims:

\begin{claim}\label{claim-single-loose}
For any allocation, every bidder $i$ has at most one weakly loose item.
\end{claim}
\begin{proof}
Let $j$ be \hlhf{a} loose item and $T$ be the set of tight items held by bidder $i$. We know \hlhf{that} $j$ is weakly loose and therefore $v_i(\{j\})+T)=b_i$. Thus the marginal value of any additional item is $0$, and therefore any other item $j'\notin \{j\}+T$ is strongly loose.
\end{proof}

\begin{claim}\label{claim-steal-tight}
If item $j$ that is tight for bidder $i$ (given the current allocation) is stolen then $i$ will not steal item $j$ later.
\end{claim}
\begin{proof}
After item $j$ is stolen by $i'$, the new price of $j$ is greater than
$v_i(\{j\})$. Since by submodularity for every bundle $S$, $v_i(\{j\}) \geq v_i(j|S)$, and since the price of $j$ only goes up, $i$ never steals item $j$ from any other bidder later in the algorithm.
\end{proof}

\begin{claim}\label{claim-steal-loose}
If item $j$ that is weakly loose for $i$ is stolen then $i$ will not steal item $j$ again during the algorithm before some tight item $j'\in S_i$ was stolen from him.
\end{claim}
\begin{proof}
The marginal value of $j$ for $i$ does not increase unless some item $j'\in S_i$, $j'>_ij$, is stolen from bidder $i$. Since $j$ is weakly loose, every other item $j'$ 
\hlhf{preceding $j$ by $>_i$} is tight.
\end{proof}

\begin{claim}\label{claim-steal-strongly-loose}
If item $j$ that is strongly loose for $i$ is stolen then $i$ will not steal item $j$ again during the algorithm before some tight or weakly loose item $j'\in S_i$ was stolen from bidder $i$.
\end{claim}
\begin{proof}
The marginal value of $j$ for $i$ does not increase unless some item in $j'\in S_i$, $j'>_ij$, is stolen from him. Observe that since $j$ is strongly loose, for the marginal value to change, item $j'$ cannot be strongly loose as well.
\end{proof}

We are now ready to give an upper bound on the number of steals. An easy
corollary of Claim \ref{claim-steal-tight} is that there are at most $n\cdot m$
steals of tight items.  The next claim shows that at least one of every $O(nm)$ consecutive steals must be a steal of a tight item, and this gives us an upper bound of $(n^2\cdot m^2)$ on the total number of steals.

\begin{claim}
Every sequence of $O(nm)$ steals must contain at least one steal of an item that is
tight.
\end{claim}

\begin{proof}
The price of each item only increases, so a strongly loose item, once stolen,
will never be a strongly loose item again.  Therefore altogether there may be at
most $m$ steals of strongly loose items. We will consider from now on only
steals of weakly loose and tight items.  Among every $nm + 1$ such steals, there
is at least one bidder who lost the same item twice; if in both occasions the
item was weakly loose, by Claim~\ref{claim-steal-loose}, the bidder has
stolen this item back only because she has lost a tight item.  This shows
that at least one of these $nm+1$ steals was one of a tight item.
\end{proof}




\section{Missing Proofs of Section \ref{sec:xos}}

\subsection{The Example: Exponential Convergence of the XOS Algorithm}\label{app-xos-example}

Let $m$ be odd and let $m'=\lfloor m/2 \rfloor$. We now present a family of
submodular valuations that is parametrized by a family of \hlhf{positive} constants $k^i_S$, for every bidder $i$ and $S$ such that $|S|=m'+1$ and some $\epsilon>0$. 
Each member of this family is defined as follows:
$$
v_i(S)=
\begin{cases}
|S|& \hbox{if }|S| \leq m' ;\\
m'+\frac 1 2 +k^i_S \cdot \epsilon & \hbox{if }|S|=m'+1;\\
m'+1& \hbox{if }|S| \geq m'+2.
\end{cases}
$$

Notice that $v_i$ is indeed a submodular function as long as $k_S^i \cdot
\eps < \tfrac 1 2$, which is easy to guarantee.
Obviously, the key to getting an exponential path is choosing the $k^i_S$'s
carefully. To simplify the presentation, we will choose the $k^i_S$'s
\hlhf{as we proceed}.

We start with an allocation $(S_1,S_2)$ such that $|S_1|=m'$ and
$|S_2|=m-m'=m'+1$. Set $k^i_{S_1}=0$ for $i=1,2$. 
Bidder $1$ now bids $b_1(j)=1$ for every
item $j\in S_1$, except a single item $j'\in S_1$ where $b_1(j')=1/2$. The bid
on the rest of the items is $0$. Notice that this is a valid XOS clause of
$S_1$, and hence a \hlhf{legitimate} step in the algorithm. Now set $k^2_{S_2+j'}=1$. Notice that the best response of bidder $2$ is to take $S_2+j'$ (regardless of the choice of the other $k^2_S$'s).

The allocation now is \hlhf{$(S'_1,S'_2)=(S_1-j', S_2+j')$}. Choose some
$j''\neq j'$ in $S'_2$. Set $k^2_{S'_2}=2$. Let $b_2(j)=1$ for every item $j\in
S_2$ other than $j''\in S'_2$ and $b_2(j'')=1/2+2\epsilon$ for $j''$. The bid on
the rest of the items is $0$ again. Observe that, \hlhf{as long as we set
$k_{S'_1 + j''}^1 = 3$,} the best response of player~$1$ is to take $S'_1+j''$.
It should be clear that we can iterate this construction
to produce an exponential path such that the two
players keep exchanging a single item in each best response.  \hlhf{To be
specific,} we can use a monotone Gray code~\cite{SW95} to explicitly construct such a path (we only use the numbers that have hamming weight between $m'$ and $m'+1$). The zeros in every codeword represent items allocated to player $1$, and the ones represent items allocated to player $2$.

\subsection*{Proof of Proposition \ref{prop-xos-equiv}}
We divide the analysis into cases, according to the size of $S$:
\begin{enumerate}
\item $|S|\leq m'-20$: $C_{j}(S)=m'-20$, while for every bundle $S'$ and item $j$, $M_{S',j}(S)=A_{S'}(S)\leq |S|\leq m'-20$ and $B_{S'}(S)=\frac {m'+1} {m'+10}\cdot |S|<m'-20$.

\item $m'-20 < |S|\leq m'$: $A_{S}(S)=|S|$. On the other hand, for every other bundle $S'$ and item $j'$, $C_{j'}(S)=m'-20$, $B_{S'}(S)=\frac {m'+1} {m'+10}\cdot |S|<|S|$, and $M_{S',j'}(S)\leq A_S(S)$.

\item $m'+1 \leq |S| \leq m'+9$: let $T\subseteq S$ be a maximum value bundle of
size $m'+1$. Observe that $v_i(T)=v'_i(T)$. There exists some $j\in T$
such that $M_{T,j}(S)=m'+\frac 1 4+k^i_S$. However, for every other bundle $S'$
and item $j'$, $A_{S'}(S)\leq M_{S',j'}(S)\leq v_i(T)$, $C_{j'}(S)=m'-20$ and
\hlhf{$B_{S'}(S) \leq \tfrac{m' + 1}{m' + 10} \cdot (m' + 9) = m' + 1 -
\tfrac{m' + 1}{m' + 10} \leq m' + \tfrac 1 4 \leq v_i(S)$}.

\item $|S|\geq m'+10$: for some bundle $S'\subseteq S$, $|S'|=m'+10$ we have that $B_{S'}(S)=m'+1$, and this value is larger than the maximum value of every other additive valuation in the support of $v'_i$.
\end{enumerate}

\subsection*{Proof of Proposition \ref{prop:eq-char}}

We divide again into several cases. We only consider cases where $|S_1|\leq m'$, otherwise we have that $|S_2|\leq m'$ and the proof is symmetric.
\begin{enumerate}
\item $|S_1|\leq m'-20$: in this case, \hlhf{$v_1(S_1) = m' - 20$ and} $|S_2|=m-|S_1| \geq m-m'+20 > m'+10$.
Hence, $v(S_2)=m'+1$. In particular, bidder $2$ bids $0$ on $m-m'-10$ items
\hlhf{(since we are at a \emph{traditional} equilibrium and in the standard XOS representation at most $m'+10$ items get a positive value)}. In this case bidder $1$ is better off bidding $\epsilon>0$ on each of these items and $0$ on all other items, since his value will increase to at least $m-m'-10$ and his payment will be $0$.

\item $m'-20 < |S_1|\leq m'$: here bidder $1$ bids $1$ on every item $j\in S_1$.
Bidder $2$ is using either a $B_S$ clause or an $M_{S,j}$ clause. If bidder $2$
uses a $B_S$ clause then the price of every item in that clause is less than
$1$. Furthermore, $|S_2|\geq m'+10$ and therefore $|S_1|\leq m-m'-10$. This
implies that if bidder $1$ increases his bid to $1$ for any single item in
$S_2$: his value will increase by $1$ while his payment will increase by a
smaller amount. Notice that this is still a no-overbidding strategy for bidder
$1$. 

Hence, \hlhf{the only possibility is} that bidder $2$ uses an $M_{S,j}$ clause.
Without loss of generality $|S_2|=m-m'$ (otherwise there are some items in $S_2$
that bidder $2$ bids $0$ on). Assume that \hlhf{$S_2$} is not a $j$-local maximum. Denote $S'_1=M-S_2+\{j\}$ and observe that $v(S'_1)>v(S_2)$. Notice that Bidder $1$ can increase his profit by taking $S'_1$: his profit in this case is $v(S'_1)-p_j=v(S'_1)-(v(S_2)-m')> m'$, and it is easy to see that $v(S_1)\leq m'$ since $|S_1|\leq m'$.
\end{enumerate}

\subsection*{Proof of Lemma \ref{lemma-bounded-cover}}
Let $S$ be a bundle covered by $\dq$, and let $v$ be a sensitive
valuation such that $\dq(v)=S$. Denote the price of every item $j$ in $\dq$ by $p_j$. Let $t\in\argmin_j p_j$. Notice that $p_t\leq \frac 1 2$, otherwise any bundle of size $m'$ is more profitable than $S$.  Since the profit from $\{t\}$  is $m'-20-p_t$ and yet the demand query returns $S$ it holds that: 
\begin{eqnarray}
m'+\frac 1 4 + k_S - \sum_{j \in S} p_j &\geq& m'-20-p_t \nonumber\\ 
m'+\frac 1 4 + k_S - \sum_{j\in S} p_j &\geq& m'-20-\frac 1 2   \nonumber\\
21 & > &  \sum_{j\in S}p_j \label{eqn-profit}
\end{eqnarray}
\HuNote{removed the superscript of $k^i_S$, since we don't talk about bidders
here.}

Next we observe that there is no set $S'\subseteq M-S$, \hlhf{$|S'| = 9$}, such that $\Sigma_{j\in S'}p_j<\frac 1 2$. This is true since if such $S'$ exists then $S\cup S'$ is more profitable than $S$ (observe that $v(S\cup S')=m'+1$ since $|S\cup S'|=m'+10$):
\begin{align*}
v(S\cup S')-\Sigma_{j\in S\cup S'}p_j & > m'+1 - \Sigma_{j\in S}p_j - \frac 1 2
\\
& \geq v(S)-\Sigma_{j\in S} p_j
\end{align*}

Therefore, for all items $j\in M-S$ but \hlhf{at most nine} of them, it holds that $p_j\geq \frac 1 {18}$. Denote these expensive items by $E$ and let $C=M-S-E$ (i.e., the set of cheap items). 

The crux of the proof is that if some bundle $G$ is covered by $\dq$, then it
cannot contain more than $18\cdot 21=378$ items from $E$, otherwise $21 <
\Sigma_{j\in G}p_j$, a contradiction to \eqref{eqn-profit}. Thus, each
$G$ covered by $\dq$ is composed of up to $378$ items from~$E$
(of which there are at most $(|E|+1)^{378}$ possible choices, where we use
$|E|+1$ and not just $|E|$ since less than 378 items might be chosen from $E$),
up to nine items from $C$ (of which the number of possible choices is $2^9$), and additional
items from $S$ (we use the bound $(|S|+1)^{387}$ since we have to choose up to
$387 = 378+9$ items to exclude from $S$). A bound on the number of bundles that can be
constructed this way --- and hence also a bound on the number of bundles covered
by $\dq$ --- is therefore:
\begin{align*}
(|E|+1)^{378}\cdot 2^9 \cdot {(|S|+1)^{387}  }& \leq m^{378}\cdot 2^9 \cdot
m^{387} \\
& < 1000\cdot m^{765}
\end{align*}

\section{Local Maximum vs.\@ Nash Equilibrium}\label{app-PLS}

Consider combinatorial auction with submodular bidders. A \emph{local maximum} is an allocation of the items $(S_1,\ldots, S_n)$ such that for every $j\in S_i$ and player $i'$ we have that $\Sigma_kv_k(S_k)\geq v_i(S_i-\{j\})+v_{i'}(S_{i'}+\{j\})+\Sigma_{k\neq i, i'}v_k(S_k)$. Already in \cite{CKS10} it was observed that if the valuations are submodular then for every allocation $(S_1,\ldots, S_n)$ that is a local maximum there are bids that together constitute a no overbidding equilibrium (this is the reasons why the price of stability in these games is $1$: the welfare maximizing allocation is obviously a local maximum). 

One could have hoped that to find an equilibrium with a good approximation one could just find a local maximum with a good approximation. However, finding a local maximum turns out to be PLS complete. To see this, we start with the PLS complete problem of finding a maximal cut and convert it to a combinatorial auction with two submodular bidders with identical valuations. Given a weighted graph, the set of items will be the set of items and the value of a set of vertices $S$ is the weight of all edges that have at least one end point at $S$. Observe that a local maximum in the combinatorial auction corresponds to a maximal cut, and vice versa.

Interestingly, in our reduction it is easy to find a nash equilibrium in
polynomial time: this is a direct consequence of our algorithm for a constant
number of submodular bidders (the number of bidders is two in our reduction).
This in particular proves that the set of equilibria \hlhf{strictly contains} the set of local maxima (hence it is easier to find an equilibrium than a local maximum).

\end{document}